%% file: main.tex
\documentclass{IEEEtran4PSCC}

\input{Files/0Packages}

\begin{document}

\twocolumn

\title{SecuLEx: a Secure Limit Exchange Market for Dynamic Operating Envelopes}

\input{Files/Authors}

\input{Files/Abstract}
\input{Files/Introduction}
\input{Files/EnvelopeAllocation}
\input{Files/MarketProblem}
\input{Files/Implentation}
\input{Files/Discussion2}

\input{Files/Conclusion}

\printbibliography
\appendix
\input{Files/ProofSafeness_DC}

\end{document}

%% file: Files/0Packages.tex
\usepackage{amsmath,amssymb,amsfonts,bbold,amsthm}
\usepackage{blkarray} 
\usepackage{hyperref}
\usepackage{algpseudocode}
\usepackage[normalem]{ulem}
\usepackage{enumitem}
\algrenewcommand\algorithmicindent{0.4em}
\usepackage{multirow} 

\usepackage{pifont}
\usepackage{tikz}
\usetikzlibrary{arrows.meta, positioning}

\usepackage{wasysym}
\usepackage{array}

\usepackage{booktabs} 
\newtheorem{theorem}{Theorem}

\usepackage{siunitx} 

\usepackage[dvipsnames]{xcolor} 

\usepackage{tcolorbox}
\usepackage{amsmath, amssymb}
\usepackage{enumitem}

\usepackage{makecell}

\usepackage{scalerel}

\usepackage{caption}
\usepackage[sorting=none]{biblatex}

%% file: Files/Authors.tex
\author{

    \IEEEauthorblockN{
        Maurizio Vassallo$^{a*}$,
        Adrien Bolland$^{a*}$,
        Alireza Bahmanyar$^b$,
        Louis Wehenkel$^{a}$,
        Laurine Duchesne$^b$,
        Dong Liu$^c$, \\
        Sania Khaskheli$^d$,
        Alexis Ha Thuc$^e$,
        Pedro P. Vergara$^c$,
        Amjad Anvari-Moghaddam$^d$,
        Simon Gerard$^f$,
        Damien Ernst$^a$
    }
    
    \IEEEauthorblockA{
        $^a$University of Liège, Belgium,\\
        $^b$Haulogy, Belgium, \\
        $^c$Delft University of Technology, The Netherlands,\\
        $^d$Aalborg University, Denmark,\\
        $^e$INSA Lyon, France,\\
        $^f$RESA, Belgium.
    }
}

\maketitle

\def\thefootnote{*}\footnotetext{These authors contributed equally to this work}\def\thefootnote{\arabic{footnote}}

%% file: Files/Abstract.tex
\begin{abstract}
Distributed energy resources (DERs) are transforming power networks, challenging traditional operational methods, and requiring new coordination mechanisms.
To address this challenge, this paper introduces SecuLEx (Secure Limit Exchange), a new market-based paradigm to allocate power injection and withdrawal limits that guarantee network security during time periods, called dynamic operating envelopes (DOEs).
Under this paradigm, distribution system operators (DSOs) assign initial DOEs to customers. 
These limits can be exchanged afterward through a market, allowing customers to reallocate them according to their needs while ensuring network operational constraints.
We formalize SecuLEx and illustrate DOE allocation and market exchanges on a small-scale low-voltage (LV) network, demonstrating that both procedures are computationally tractable. 
In this example, SecuLEx reduces renewable curtailment and improves grid utilization and social welfare compared to traditional approaches.
\end{abstract}



%% file: Files/Introduction.tex
\section{Introduction}
\label{introduction}
The rapid penetration of distributed energy resources (DERs), such as photovoltaic (PV) systems, electric vehicles (EVs), and battery energy storage systems (BESSs), is transforming the operation of electricity distribution networks.
Originally designed for centralized generation and relatively predictable, unidirectional power flows, today’s networks are being pushed toward more volatile and decentralized energy generation.
This challenges the secure operation of distribution networks, consisting of avoiding large voltage deviations and line congestion, potentially compromising their reliability.
For instance, the widespread adoption of rooftop PV can lead to reverse power flows and overvoltage issues \cite{PVproblems2020}. 
At the same time, the growing popularity of EVs may worsen network congestion, particularly as peak EV charging tends to coincide with the existing evening peak \cite{EVproblems2018}.
These developments necessitate a new generation of coordination mechanisms to ensure the secure and reliable integration of DERs \cite{ReviewDERChallenges2021}.

Recent reports estimate that renewable energy curtailment due to network constraints costs the European Union between €1.5 and €2.5 billion annually \cite{Thomassen2024Curtailment}.
Historically, distribution networks were built under a “fit and forget” philosophy, where infrastructure was sized to accommodate static worst-case consumption scenarios. Such designs were based on unidirectional power flows and stable demand patterns. However, the high magnitude of bidirectional flows introduced by emerging DERs was not anticipated. As a result, even oversized networks often fail to integrate new DERs efficiently, leading to voltage violations, thermal overloads, and mismatches between generation and demand.

To solve these network issues, different centralized control strategies have been proposed. 
Active network management (ANM) allows system operators to dynamically curtail PV injections, coordinate flexible loads, or schedule DERs to reduce the stress on the grid \cite{Gemine2017}. 
Building on these ANM principles, specific strategies such as controlled EV charging \cite{EvAnm2014} have been developed to enhance network hosting capacity (HC) and avoid extensive grid reinforcement, complementing traditional approaches like targeted infrastructure upgrades \cite{Meunier2021}.
However, these centralized methods face practical limitations in scalability and adaptability in networks with high DER penetration and limited observability \cite{ZHANG2024110076}.

To address the limitations of centralized control, the concept of dynamic operating envelopes (DOEs) has emerged as a potential solution \cite{DOEs2024}.
DOEs are dynamic power limits on customer injections and withdrawals that ensure network security.
Various methods have been proposed for DOE computation. Authors in \cite{Petrou2021} developed a dynamic approach for assigning export limits to PV devices, explicitly accounting for network constraints and fairness considerations. Authors in \cite{vergara2023mixed} proposed a mixed-integer linear programming (MILP) formulation to compute customer-specific limits that balance network feasibility and fairness. Similarly, \cite{gupta2024} proposed a day-ahead stochastic optimization framework for computing fair PV limits using a linear approximation of the AC power flow model.
While these approaches ensure network security and fair allocation, they treat DOEs as fixed and non-transferable, limiting the ability to fully leverage the network’s flexibility under high DER penetration.

In parallel, decentralized market mechanisms have emerged as tools for coordinating DERs in distribution systems. Local energy markets (LEMs) enable peer-to-peer (P2P) trading of energy, facilitating localized balancing and reducing reliance on central infrastructure \cite{LEM2019, Liu2024ASE}. 
Other works explore market-clearing mechanisms that incorporate voltage and network constraints directly into auction models \cite{Market2015}, or game-theoretic frameworks to manage uncertainties in networks with high integration of DERs \cite{dolatabadi2024stochastic}.
Market mechanisms have demonstrated efficiency in terms of cost-effectiveness and improved network utilization in energy trading, but their potential for DOE allocation has so far remained largely unexplored \cite{SCOPE2025}.


This paper presents SecuLEx, a novel market-based framework that enables customers to trade their DOEs while ensuring network security. 
Unlike existing static approaches, SecuLEx provides a framework through which customers can reallocate limits according to their needs, offering a possible improvement in the use of existing network capacity.
The main contributions of this work are the following:
\begin{itemize}

    \item We propose a novel DOE allocation mechanism based on a lexicographic max–min formulation, which guarantees fairness across customers and serves as the initial allocation before the market exchanges.

    \item We define the SecuLEx market, including order specifications, market clearing with security verification, and settlement mechanisms.

\end{itemize}

The rest of the paper is organized as follows. 
Section \ref{envelopes} presents the mathematical foundation for grid operation using DOEs, including network representation, DOE definitions, security verification mechanisms, and DOE allocation. 
Section \ref{blues_market} introduces the SecuLEx market organization, covering temporal market structure, order specifications, and the clearing and settlement mechanisms. 
Section \ref{implementation} provides a concrete implementation and illustration of SecuLEx for a radial LV network.
Section \ref{discussion} discusses key design choices and limitations of the proposed approach. 
Section \ref{conclusion} concludes the paper and outlines future research directions.

%% file: Files/EnvelopeAllocation.tex
\section{Grid operation with envelopes}
\label{envelopes}
This section establishes the mathematical foundation for operating distribution networks using DOEs. 
Section \ref{ss:network} presents the network model.
Section \ref{ss:does} formalizes the concept of DOEs.
Section \ref{ss:verification} introduces the security verification function that checks network security under the given DOEs. 
Section \ref{ss:CalculateLimitsFunction} formulates the optimization for DOE allocation.

\subsection{Network representation}
\label{ss:network}
An electrical network is modeled as a directed graph, denoted by $\mathfrak{G} = (\mathcal{N}, \mathcal{E})$, where $\mathcal{N}$ represents the set of nodes and $\mathcal{E} \subset \mathcal{N} \times \mathcal{N}$ represents the set of directed edges. 
Each node $n \in \mathcal{N}$ corresponds to a network component, such as a transformer or a customer. Each edge $e \in \mathcal{E}$ establishes a connection between two nodes of the set $\mathcal{N}$ and is commonly referred to as a line. The orientation is arbitrarily fixed by convention to simplify later modeling.

\subsubsection{Network constraints}
Each node $n \in \mathcal{N}$ has an associated voltage $V_n \in \mathbb{C}$, and $V \in \mathbb{C}^{|\mathcal{N}|}$ is the voltage vector composed of the voltages at each node $n \in \mathcal{N}$.
The voltage magnitude at each node $n \in \mathcal{N}$ must satisfy the voltage magnitude constraint $\underline{V} \leq | V_n |\leq \overline{V}$, where $\underline{V}, \overline{V} \in \mathbb{R}_+$ are the minimum and maximum allowable voltage magnitudes.
Similarly, each line $e \in \mathcal{E}$ has an associated current flowing through the line $I_e \in \mathbb{C}$, and $I \in \mathbb{C}^{|\mathcal{E}|}$ is the current vector composed of the current flows at each line $e \in \mathcal{E}$.
Each line $e \in \mathcal{E}$ must satisfy the current magnitude constraint $|I_{e}| \leq \overline{I}$, where $\overline{I} \in \mathbb{R}_+$ is the maximum current allowed through a line.


\subsubsection{Customers' power usage}
Let $\mathcal{C} \subset \mathcal{N}$ be the set of customer nodes, i.e., nodes where power is withdrawn or injected.
Let $S_{c} = P_{c} + j Q_{c} \in \mathbb{C}$ denote the complex power consumption of $c \in \mathcal{C}$, composed of an active part $P_{c} \in \mathbb{R}$ and a reactive part $Q_{c} \in \mathbb{R}$.
By convention, positive power values correspond to withdrawal and negative values to injection.


\subsection{DOEs definition}
\label{ss:does}
The DSO assigns DOEs, specifying the upper and lower limits of active and reactive power each customer is allowed to consume. These limits are dynamic, meaning that they can change over time, e.g., through trading or through a reallocation imposed by the DSO. 
When all customers respect their limits, the network is guaranteed to operate within secure voltage and current limits.
DOEs are specified by the complex matrix $\mathbf{L}$ of dimension $|\mathcal{C}| \times 2$, defined as:
\begin{align}
    \mathbf{L} = [\underline{S}, \overline{S}] =
    \begin{blockarray}{cc}
        \begin{block}{[cc]}
        \underline{S}_{\, 1}  & \overline{S}_{1}   \\
        \vdots  & \vdots   \\
        \underline{S}_{\, |\mathcal{C}|}  & \overline{S}_{|\mathcal{C}|}   \\
        \end{block}
    \end{blockarray} ,
\end{align}
where $\underline{S}_c =  \underline{P}_c + j \underline{Q}_c$ and $\overline{S}_c = \overline{P}_c + j \overline{Q}_c$ represent the lower and upper limits for customer $c \in \mathcal{C}$.
We write $S = P + jQ \in \mathbf{L}$ when $P$ and $Q$ both respect the bounds defined in the DOE matrix for each customer.

\subsection{Security verification function}
\label{ss:verification}
To ensure secure network operations, we introduce a verification function $VerifyLimits$ that determines whether a given DOE allocation preserves network security. 
The function returns a non-positive value if and only if the limit matrix $\mathbf{L}$ ensures secure operation:
\begin{equation}
\label{eq:VerifyLimitsGeneric}
\begin{cases}
    VerifyLimits(\mathfrak{G}, \mathbf{L}) \leq 0 & \forall S \in \mathbf{L}: \mathfrak{G} \text{ is secure}, \\
   VerifyLimits(\mathfrak{G}, \mathbf{L}) > 0 & \text{otherwise}.
\end{cases}
\end{equation}
Evaluating this function requires solving the power flow:
\begin{align}
    V, I = PowerFlow(\mathfrak{G}, S) && \forall S \in \mathbf{L} \; , \label{eq:PowerFlow}
\end{align}
and verifying that the voltage limit $\underline{V} \leq |V| \leq \overline{V}$ and current limit $|I| \leq \overline{I}$ are respected $\forall S \in \mathbf{L}$. 

\subsection{DOEs allocation}
\label{ss:CalculateLimitsFunction}
DOE allocation is formulated as a constrained optimization problem that maximizes network utility while guaranteeing secure operation for any power profile within the assigned limits.  
This optimization problem is formulated as follows:
\begin{subequations}
\label{eq:CalculateDOEs}
\begin{align}
    \max_{\mathbf{L}} \quad & U(\mathbf{L}), \\
    \text{s.t.} \quad & VerifyLimits(\mathfrak{G}, \mathbf{L}) \leq 0 \; ,
\end{align}
\end{subequations}
where $U$ is a general utility function representing the DSO's objective, such as maximizing customer flexibility, fairness, or system efficiency. The function $VerifyLimits$ ensures that, for all power injections within limits, the network remains within secure operational boundaries.

%% file: Files/MarketProblem.tex
\section{SecuLEx market organization}
\label{blues_market}
This section presents the market framework that enables customers to exchange their DOEs.
Section \ref{ss:temporaloperation} defines the temporal organization of trading periods.
Section \ref{ss:orderbook} specifies the order structure and book management for buy and sell requests of envelope portions.
Section \ref{ss:clearing} introduces the clearing mechanism and Section \ref{ss:settlement} the settlement function that determines the money exchanges between customers.

\subsection{Limit market temporal operation}
\label{ss:temporaloperation}
In SecuLEx, DOEs are assigned in advance for future time periods, called \textit{products}. For example, the DSO may assign limits one day in advance for all hour intervals of the following day. A product is identified by its starting time $t \in \mathcal{T}$ (e.g., $\mathcal{T} = \{\text{00:00} \, , \text{01:00} \, ,\text{02:00} \, , \ldots \, , \text{23:00}\}$ for 24-hour intervals in a day). 
For each product $t \in \mathcal{T}$, the exchange window opens immediately after the DOE assignment and closes shortly before the corresponding time period. This allows participants to adjust their limits to reflect updated forecasts or preferences.

Figure~\ref{fig:marketTimeline} illustrates the temporal structure of a single product: the DSO assigns DOEs, the continuous exchange window opens, the order book fills, and trading continues until market closure just before real-time operation. 
\begin{figure}[h!]
\centering
\includegraphics[width=0.49\textwidth]{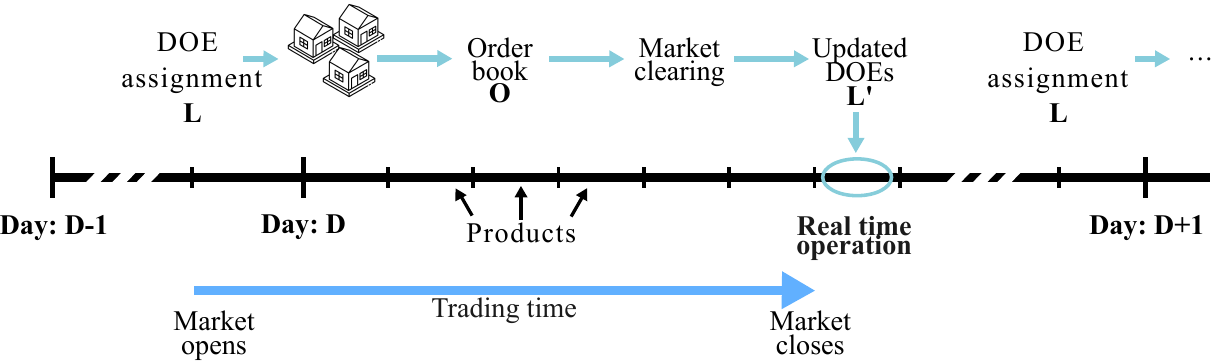}
\caption{Timeline illustrating the sequence from the initial DOE assignment by the DSO to the closing of the continuous exchange window just before the operating period.}
\label{fig:marketTimeline}
\end{figure}

\subsection{Orders and order book}
\label{ss:orderbook}
Customers interact with the market by submitting orders to buy and sell portions of lower and upper limits. 
Each order expresses a customer's willingness to change their DOE limits.
In practice, customers may ask for additional upper or lower limits to gain more withdrawal or injection capacity, respectively. 
Similarly, customers may bid part of their upper or lower limits if they do not expect to use their full withdrawal or injection capacity, respectively. 
Formally, each order $o$ is defined as the tuple:
\begin{align}
    o = (id, \; c, \; type, \; bound, \; power, \; \Delta, \; \pi, \; t),
\end{align}
where:
\begin{itemize}
    \item $id \in \mathbb{N}$ is a unique identifier,
    \item $c \in \mathcal{C}$ is the customer submitting the order,
    \item $type \in \{\text{buy}, \text{ sell}\}$ distinguishes demand from supply,
    \item $bound \in \{\text{lower}, \text{ upper}\}$ indicates whether the order concerns the lower or upper limit of the DOE,
    \item $power \in \{\text{active}, \text{reactive} \}$ is the nature of the change,
    \item $\Delta \in \mathbb{R}_+$ is the requested change in limit (kW or kVAR).
    \item $\pi \in \mathbb{R}$ is the offered price (€/kW or €/kVAR),
    \item $t \in \mathcal{T}$ is the target time period (i.e., product time),
\end{itemize}
Orders are stored in the order book $\mathcal{O}$, which we formalize as the set of order identifiers.
We define $\mathcal{S}$ and $\mathcal{B}$ as the sets of identifiers for the sell and buy orders. 
Similarly, $\mathcal{L}$ and $\mathcal{U}$ are the sets of identifiers for orders on lower and upper limits.

\subsection{Clearing function}
\label{ss:clearing}
Given the limit matrix $\mathbf{L}$ and the order book $\mathcal{O}$, a market operator clears the market by matching buy and sell orders. 
The clearing process determines which orders are accepted and updates both the DOE limits and the order book accordingly. This process is formalized as:
\begin{equation}
\label{eq:ClearGeneric}
 \mathbf{L}', \mathcal{O}' = Clear(\mathbf{L}, \mathcal{O})\;,
\end{equation}
where $\mathbf{L}'$ is the updated limit matrix reflecting all accepted orders and $\mathcal{O}'$ is the updated order book containing unmatched and partially matched orders.

The updated limit matrix $\mathbf{L}'$ must preserve network security:
\begin{equation}
    VerifyLimits(\mathfrak{G}, \mathbf{L}') \leq 0\;.
\end{equation}
We assume the market operator clears orders instantaneously and publishes updated DOEs immediately, enabling continuous trading and real-time adjustment of operating limits throughout the trading window.

\subsection{Settlement function}
\label{ss:settlement}
Once orders are executed and limits updated, a settlement process determines the financial compensation for each customer.  
We define the settlement function as:
\begin{equation}
    \Pi = Settlement(\mathbf{L}, \mathbf{L}', \mathcal{O}, \mathcal{O}')\;,
    \label{eq:settlement}
\end{equation}
where $\Pi \in \mathbb{R}^{|\mathcal{C}|}$ is the resulting payment vector composed of the price compensations $\Pi_c$ for each customer $c \in \mathcal{C}$.
A positive $\Pi_c$ indicates a payment by the customer $c \in \mathcal{C}$, and a negative $\Pi_c$ indicates a payment to the customer $c \in \mathcal{C}$. 

Similarly to the clearing function, we assume that the settlement is computed instantaneously after clearing and that the payment is due to the market operator. 

%% file: Files/Implentation.tex
\section{Concrete implementation of SecuLEx}
\label{implementation}
This section provides a practical implementation of SecuLEx that achieves computational efficiency while preserving security guarantees.
Section \ref{ss:gridoperation} describes a fair DOE allocation mechanism.
Section \ref{ss:marketimplementation} describes a computationally efficient clearing mechanism for maximizing social welfare while ensuring trades preserve network security. 
Section \ref{ss:example} provides an illustrative example.

\subsection{Grid operation implementation}
\label{ss:gridoperation}
\subsubsection{Security verification function}
without additional assumptions, the security verification function in Eq.~\ref{eq:VerifyLimitsGeneric} requires solving a power flow for each $S \in \mathbf{L}$ and verifying that the voltages and currents respect the security constraints.

In this implementation, we consider a radial network and adopt a DC power flow approximation of the power flow Eq.~\ref{eq:PowerFlow}.
In this approximation, voltages take constant real values, and the complex parts of currents and powers equal zero.
Security is then expressed as constraints on the active power flow through lines, rather than voltage and current constraints. DOEs are also reduced to limits on active power and $\mathbf{L} = [ \underline{P} \, , \overline{P} ]$.
Under these assumptions, we prove in Appendix \ref{app:proof} that the security verification function in Eq. \ref{eq:VerifyLimitsGeneric} simplifies to confirming that the network $\mathfrak{G}$ remains secure only for the two boundaries, $\underline{P}$ and $\overline{P}$, of the DOE matrix:
\begin{align}
\label{eq:VerifyLimits}
 \begin{cases}
  VerifyLimits(\mathfrak{G}, \mathbf{L}) \leq 0 & \text{if both $\mathfrak{G}$ secure for $\underline{P}$} \; ,  \\
    & \phantom{\text{if both }}\text{$\mathfrak{G}$ secure for $\overline{P}$} \; ; \\
  VerifyLimits(\mathfrak{G}, \mathbf{L}) > 0 & \text{otherwise}.
\end{cases}
\end{align}

\subsubsection{DOE allocation} we assume that the DSO wants to maximize the envelope sizes while guaranteeing fairness, in the sense that the size of the smallest envelopes is maximized first. Such a multi-objective function can be formulated as a (complex) maximum utility problem as in Eq. \ref{eq:CalculateDOEs}, which eventually corresponds to solving the following (simpler) lexicographic max-min optimization problem \cite{behringer1977lexicographic}:
\begin{subequations}
\label{eq:CalculateDOEs_LV}
\begin{align}
    \underset{\mathbf{L}}{\text{lex max min}} \quad & | \underline{P}_c - \overline{P}_c| && \forall c \in \mathcal{C} \label{eq:optFunc} \; , \\
    \text{s.t.} \quad
    & VerifyLimits(\mathfrak{G}, \mathbf{L}) \leq 0 \; , \label{eq:security_constraint} \\
    & \underline{P}^{\; contr}_c \leq \underline{P}_{c} \leq \underline{P}^{\; guar}_c
    && \forall c \in \mathcal{C} \; , \\
    & \overline{P}^{\; guar}_c \leq \overline{P}_{c} \leq \overline{P}^{\; contr}_c
    && \forall c \in \mathcal{C} \; , \\
    & \underline{P}_c \!\le\! \overline{P}_c && \forall c \in \mathcal{C} \; .
\end{align}
\end{subequations}
The parameters $\underline{P}^{\; guar}_c, \overline{P}^{\; guar}_c$ denote minimum guaranteed withdrawal/injection limits, while $\underline{P}^{\; contr}_c, \overline{P}^{\; contr}_c$ denote the contractual extrema based on equipment capacity.

Complex in appearance, the problem in Eq. \ref{eq:CalculateDOEs_LV} can easily be solved by successively maximizing the size of the smallest envelope that can still be enlarged. Let $\mathcal{C'} \subseteq \mathcal{C}$ be the set of customers whose envelopes are already fixed to $[\underline{P}_{c'}^*, \overline{P}_{c'}^*]$, for all $c' \in \mathcal{C'}$. New envelopes are iteratively computed by solving:
\begin{subequations}
\label{eq:CalculateDOEs_LV_Iterative}
\begin{align}
    \max_{w, \mathbf{L}} \quad & w \; , \\
    \text{s.t.} \quad
    & VerifyLimits(\mathfrak{G}, \mathbf{L}) \leq 0 \; , \\
    & \underline{P}^{\; contr}_c \leq \underline{P}_{c} \leq \underline{P}^{\; guar}_c
    && \forall c \notin \mathcal{C'} \; , \\
    & \overline{P}^{\; guar}_c \leq \overline{P}_{c} \leq \overline{P}^{\; contr}_c
    && \forall c \notin \mathcal{C'} \; , \\
    & \underline{P}_{c} = \underline{P}_{c}^*
    && \forall c \in \mathcal{C'} \; , \\
    & \overline{P}_{c} = \overline{P}_{c}^*
    && \forall c \in \mathcal{C'} \; , \\
    & w \leq \overline{P}_c - \underline{P}_c
    && \forall c \notin \mathcal{C'} \; , \label{eq:infinity_norm_constraint}\\
    & \underline{P}_c \!\le\! \overline{P}_c && \forall c \in \mathcal{C} \; .
\end{align}
\end{subequations}
The solution to the optimization problem is the DOE matrix $\mathbf{L}^*$ and the width $w^*$. 
The customer(s) for which the constraint Eq. \ref{eq:infinity_norm_constraint} is tight, i.e., $\overline{P}_c^* - \underline{P}_c^* = w^*$, have their envelopes fixed at the current values $\mathbf{L}_c^* = [\underline{P}_c^*, \overline{P}_c^*]$. 
These customers are added to $\mathcal{C}'$, and the procedure repeats until all customers have been assigned fixed envelopes.

\subsection{SecuLEx market implementation} 
\label{ss:marketimplementation}

\subsubsection{Clearing function} the clearing function implements a centralized market where orders are accepted to maximize the social welfare. Here, accepted orders lead to updated envelopes that are constrained by the security verification function. Assuming orders can be partially cleared, we compute for each order $o \in \mathcal{O} = \mathcal{B} \cup \mathcal{S}$ the accepted quantity $a_o$ by solving:
\begin{subequations}
\label{eq:market-clear}
\begin{align}
    \max_{a, \mathbf{L}'} \quad & \sum_{b\in\mathcal{B}} \pi_b a_b - \sum_{s\in\mathcal{S}} \pi_s a_s \; , \\
    \text{s.t.}\quad
        & 0 \le a_o \le \Delta P_o && \forall o \in\mathcal{O} \; ,  \\
        & \underline P_c' = \underline P_c + \! \sum_{\substack{s \in \mathcal{S} \cap \mathcal{L}: \\ c_s = c}} \! a_s - \! \sum_{\substack{b \in \mathcal{B} \cap \mathcal{L}: \\ c_b = c}} \! a_b && \forall c \in \mathcal{C} \; , \\
        & \overline P_c' = \overline P_c - \! \sum_{\substack{s \in \mathcal{S} \cap \mathcal{U}: \\ c_s = c}} \! a_s + \! \sum_{\substack{b \in \mathcal{B} \cap \mathcal{U}: \\ c_b = c}} \! a_b && \forall c \in \mathcal{C} \; , \\
        & \underline{P}_c' \!\le\! \overline{P}_c' && \forall c \in \mathcal{C} \; , \\
        & \mathbf{L}' = [\underline{P}', \overline{P}'] \; , \\
        & VerifyLimits(\mathfrak{G}, \mathbf{L}') \leq 0 \; .
\end{align}
\end{subequations}
Solving this optimization problem provides the updated secure limits $\mathbf{L}'$ and the updated order book $\mathcal{O}'$ containing the uncleared quantities.

\subsubsection{Settlement function}
orders are settled on a pay-as-bid settlement mechanism. 
Each customer's total payment equals the sum of their individual order payments, where each order payment is the order price multiplied by the cleared quantity:
\begin{align}
\label{eq:settlement_LV}
\Pi_c &= \sum_{b \in \mathcal{B}: c_b = c} a_b \pi_b - \sum_{s \in \mathcal{S}: c_s = c} a_s \pi_s \; .
\end{align}

The pay-as-bid settlement leads to a surplus that equals the social welfare $SW$ and that is captured by the market operator:
\begin{align}
    SW = \sum_{c \in \mathcal{C}} \Pi_c = \sum_{b\in\mathcal{B}} \pi_b a_b - \sum_{s\in\mathcal{S}} \pi_s a_s \; .
\end{align}

\input{Files/ExampleScenarios}

%% file: Files/ExampleScenarios.tex
\subsection{Illustrative example}
\label{ss:example}
This section illustrates SecuLEx on a simple radial low-voltage (LV) distribution network.
The illustration follows the temporal structure outlined in Fig.~\ref{fig:marketTimeline}, considering the DOE assignment on Day $D-1$ for a single period of one hour (12:00-13:00 on Day $D$), with continuous market operation starting at DOE assignment until real-time (12:00). 
We present a scenario with congestion and demonstrate the value of SecuLEx compared to alternative management strategies.

The network is represented by the graph $\mathfrak{G} = (\mathcal{N}, \; \mathcal{E})$, with nodes $\mathcal{N} = \{T, \; B, \; C_1, \; C_2, \; C_3, \; C_4\}$ corresponding to a transformer $T$, an intermediate bus $B$, and four representative customers.
Each customer is connected to the bus $B$ that is supplied by the transformer $T$.
The interconnections are represented by the edges
$\mathcal{E} = \{(T, B), \; (B, C_1), \; (B, C_2), \; (B, C_3), \; (B, C_4)\}$.
The transformer must operate below $60$kW, which translates to a security constraint limiting the power flowing through $(T, \; B)$ to $60$kW.

Let us consider that during the day, customers have estimations of their future consumption for the midday period. Customers have the following installations, expected consumptions for that period, and known electricity prices:
\begin{itemize}
  \item $C_1$ is a load-only: net withdrawal $8$kW. They pay $0.15$€/kWh.
  \item $C_2$ has a PV installation: net injection $21$kW. They pay $0.16$€/kWh for withdrawal and get paid $0.06$€/kWh for injection.
  \item $C_3$ has a PV installation: net injection $26$kW. They pay $0.15$€/kWh for withdrawal and get paid $0.04$€/kWh for injection.
  \item $C_4$ has a PV-battery installation: PV production of $24$kW, battery capacity of $\pm 10$kW and discharging by $6$kW; net injection $30$kW. They pay $0.15$€/kWh for withdrawal and get paid $0.02$€/kWh for injection.
\end{itemize}

In this scenario, at midday, there is an expected reverse power flow from the customers to the transformer of $69$kW.
The magnitude of this power flow exceeds the transformer's $60$kW limit, causing congestion, as shown in Fig. \ref{fig:baseline_congestion}.

\begin{figure}[h!]
\centering
\includegraphics[width=0.48\textwidth]{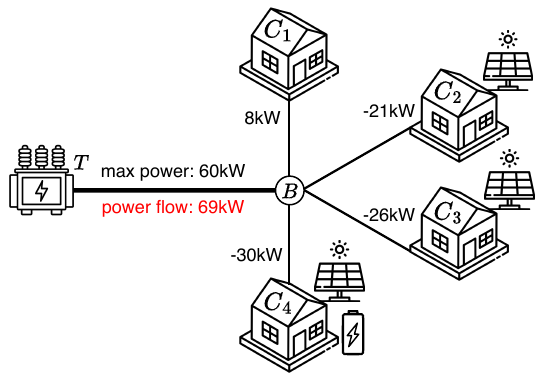}
\caption{Baseline situation where the reverse power flow ($69$kW) exceeds the transformer limit ($60$kW), causing congestion.}
\label{fig:baseline_congestion}
\end{figure}

Let us compare the outcome depending on the congestion-management scheme of the DSO, assuming that the expected customer consumption is realized at midday.

\textbf{No control.}
In this scenario, the DSO takes no action.
Without intervention, the line $(T,B)$ is overloaded by $9$kW. This causes a security issue, potentially leading to line disconnection and disrupting all customers at the cost of the DSO.

\textbf{Centralized ANM.}
This case represents a traditional management scheme where the DSO reacts in real-time when the problem is emerging.
The DSO issues curtailment commands to eliminate the overload, but cannot control customer batteries (no contractual access).

Let us assume that the DSO takes the minimal corrective actions to guarantee security and curtails PV generation equally across $C_2$, $C_3$, and $C_4$ until the limit is respected.
Thus, the total curtailed PV is $9$kW, equally divided over the three customers injecting power, and the transformer operates securely at $60$kW. Curtailment leads to an opportunity loss of $3 \cdot 0.06 + 3 \cdot 0.04 + 3 \cdot 0.02 = 0.36$€.

\textbf{Static envelope assignment.}
In this case, the DSO assigns static envelopes one day ahead (Day $D-1$), solving the lexicographic max-min allocation from Eq.~\ref{eq:CalculateDOEs_LV}.
Here, the security verification function is implemented relying on linearized (DC) power flow equations as discussed previously, leading to the simplified implementation of Eq. \ref{eq:VerifyLimits}.
 
An optimal envelope allocation is:
\begin{align*}
    \mathbf{L}_{1} &= [0,   15] \; , &
    \mathbf{L}_{2} &= [-20, 15] \; , \\
    \mathbf{L}_{3} &= [-20, 15] \; , &
    \mathbf{L}_{4} &= [-20, 15] \; .
\end{align*}
Customers are curtailed in case of envelope violation, guaranteeing secure and distributed real-time operation.

At this stage, customers without flexibility have no alternative to curtailment. Nevertheless, customers with available flexibility may decide to adapt their consumption pattern, knowing the conditions of curtailment. This choice is driven by the balance between the opportunity loss due to curtailment and the cost of changing the schedule, i.e., expected reduction of the profit from strategically charging and discharging the battery during the day.

We consider two cases, depending on Customer $C_4$:
\begin{itemize}
    \item Customer $C_4$ evaluates that changing the battery schedule may be too expensive, considering their low opportunity loss. 
    The customer maintains their planned discharge, resulting in a net injection of $30$kW, which violates their envelope and triggers $10$kW of curtailment.
    Combined with curtailment from other customers, the total curtailment reaches $17$kW during one hour, leading to an opportunity loss of $1 \cdot 0.06 + 6 \cdot 0.04 + 10 \cdot 0.02 = 0.50$€.
    
    \item Customer $C_4$ anticipates curtailment and charges their battery with $4$kW instead of injecting power, remaining within their assigned limit with a net injection of $20$kW. 
    The total curtailment is $7$kW during one hour and leads to an opportunity loss of $1 \cdot 0.06 + 6 \cdot 0.04 = 0.30$€.
\end{itemize}

In this scenario, as customers know their limits, they are incentivized to use their flexibility.
Nevertheless, static envelopes are conservative and may result, in this case, in greater curtailment compared to the centralized ANM, which only responds to the actual congestion observed in real-time.

\textbf{DOE with SecuLEx.}
Starting from the DOE allocation above, market participants can submit orders to buy/sell portions of their envelopes. In theory, this exchange is driven by their expected future energy usage and corresponding cost.

In an ideal market, each customer should submit orders to maximize their expected future profit. It includes minimizing the expected costs that may result from future curtailment, optimizing the consumption schedules depending on the hourly costs of energy, and maximizing the expected revenues from selling unused parts of their envelopes. In the hypothetical case where customers have known and inflexible future consumptions, they should ask for additional limits to avoid curtailment at the incurred marginal cost and bid unused limits at zero marginal cost. Customers with flexibility furthermore submit orders to also optimize their consumption schedules. {\color{black}In practice, uncertainty about future consumptions, market inefficiencies, and speculation are reflected in prices and volumes.

Let us illustrate the trading incentives in this context:
\begin{itemize}
    \item Customer $C_1$ is expecting to withdraw power from the network, and we assume they are not interested in entering the market.
    It may be considered a market inefficiency or may be due to the need for flexibility within their current envelope.

    \item Customer $C_2$ is expecting curtailment of $1$kW during one hour, leading to an opportunity loss of $0.06$€.
    The customer thus offers to buy $1$kW of lower limits at the price of $0.03$€/kW.

    \item Customer $C_3$ is expecting curtailment of $6$kW during one hour, leading to an opportunity loss of $6 \cdot 0.04 = 0.24$€. 
    The customer offers to buy $6$kW of lower limits at the price of $0.02$€/kW.

    \item Customer $C_4$ is expecting curtailment of $10$kW during one hour, leading to an opportunity loss of $10 \cdot 0.02 = 0.20$€, if they discharge their battery.
    On the one hand, the customer could ask $10$kW additional lower limit to avoid curtailment. 
    On the other hand, expecting that lower limits are going to be scarce on the market, they can decide to exploit their battery's flexibility to charge the battery by $10$kW. Then, the customer respects its DOE and may even bid $6$kW of lower limit. 
    Assuming this situation occurs, the order price should reflect the cost of rescheduling the battery and is assumed to be $0.02$€/kW.
\end{itemize}
All corresponding orders are represented in Table \ref{tab:orders}.
}

\begin{table}[h!] 
\centering 
\begin{tabular}{c c c c c c} 
    \toprule    ID & Customer & Type & Bound & Price [€/kW] & Quantity [kW] \\ \midrule
                $1$ & $C_2$ & Buy  & Lower & 0.03 & 1 \\
                $2$ & $C_3$ & Buy  & Lower & 0.02 & 6 \\
    \midrule    $3$ & $C_4$ & Sell & Lower & 0.02 & 6 \\
\bottomrule 
\end{tabular} 
\caption{Initial order book.} 
\label{tab:orders} 
\end{table}

Providing this order book, the market-clearing algorithm of
Eq.~\ref{eq:market-clear} can be executed and yields the following updated limits:
\begin{align*}
\label{eq:newlimits}
    \mathbf{L}_{1} &= [0,   15] \; , &  
    &\mathbf{L}_{2} = [-21, 15] \; , \\
    \mathbf{L}_{3} &= [-25, 15] \; , &
    &\mathbf{L}_{4} = [-14, 15] \; .
\end{align*}

The exchange between consumers and new envelopes corresponds to a complete clear of orders from Customer $C_2$ and Customer $C_4$, and a partial clear of the order from Customer $C_3$. The updated order book corresponds to the uncleared orders and is represented in Table \ref{tab:remaining_orders}.

Under the pay-as-bid principle, Customer $C_2$ pays $\Pi_{C_2} = 1 \cdot 0.03 = 0.03$€, Customer $C_3$ pays $\Pi_{C_3} = 5 \cdot 0.02 = 0.10$€, and Customer $C_4$ is paid $-\Pi_{C_4}= 6 \cdot 0.02 = 0.12$€.
The social welfare $SW = 0.01$€ is captured by the market operator, as a regulated fee to reduce tariffs or reinvest in the grid.

\begin{table}[h!] 
\centering 
\begin{tabular}{c c c c c c} 
    \toprule    ID & Customer & Type & Bound & Price [€/kW] & Quantity [kW] \\ \midrule
                $2$ & $C_3$ & Buy  & Lower & 0.02 & 1 \\
\bottomrule 
\end{tabular} 
\caption{Updated order book.} 
\label{tab:remaining_orders} 
\end{table}

Customers may submit additional orders as their consumption forecasts evolve, continuing until the market closure, where DOEs are frozen. 
Let us consider that no additional exchange materializes and that the expected consumptions and injections are realized in real-time. 
Customer $C_3$ is the only customer with excess injection and is curtailed by $1$kW. 
The market has created incentives for exploiting flexibility, and Customer $C_4$ has adapted its behavior to absorb some of the PV production, leading to smaller curtailment.

\textbf{Summary.} Table~\ref{tab:comparison} summarizes the performance of the different management schemes. As discussed previously, SecuLEx incentivizes the use of flexibility without requiring real-time centralized operation, leading to the lowest curtailment, highest renewable utilization, and market social welfare.
\begin{table}[!h]
\centering
\resizebox{\columnwidth}{!}{%
\begin{tabular}{l c c c c}
\toprule
\textbf{Metric} & \textbf{No Control} & \textbf{ANM} & \textbf{Static envelopes} & \textbf{SecuLEx} \\
\midrule
Curtailment [kW] & 0 & 9 & 17 or 7 & 1 \\
Renewable utilization [\%] & 100 & 87 & 76 or 90 & 99 \\
Security violation & Yes & No & No & No \\
Incentivizes flexibility & No & No & Partial & Yes \\
Market social welfare [€] & / & / & / & 0.01 \\
\bottomrule
\end{tabular}
}
\caption{Comparison of management approaches.}
\label{tab:comparison}
\end{table}
\vspace{-5mm}

%% file: Files/Discussion2.tex
\section{Discussion}
\label{discussion}
SecuLEx represents a change in network operation in two fundamental respects. First, it moves from operation strategies in reaction to security issues, such as ANM, to forward planning of envelopes that guarantee security. Second, it introduces market-driven coordination that allows responsiveness to real-time system conditions and customer preferences. This section reflects on technical aspects of the implementation and economic implications related to SecuLEx.

\subsection{Technical and computational challenges}

\subsubsection{Computational scalability} the initial envelope allocation and order clearing both require solving optimization problems, respectively, problem Eq.~\ref{eq:CalculateDOEs_LV} and Eq.~\ref{eq:market-clear}. In practice, the first problem requires solving one optimization program per customer, where the number of variables and constraints scales linearly with the number of customers. In the second problem, variables and constraints scale linearly with the number of orders. When many customers are involved, submitting many orders, the optimization problems will be time-consuming to solve. This may hinder the good working of SecuLEx. In particular, this may require rethinking the continuous operation or accepting heuristic solutions.

\subsubsection{Network security verification} in Appendix \ref{app:proof}, we show that, for radial network topology and DC power flow, the problem reduces to verifying security only at the DOE boundaries. This simplification may not hold for networks with large losses, in meshed topologies, heavily loaded feeders, or cases with advanced reactive power control. Then, ensuring security when solving optimization problems Eq.~\ref{eq:CalculateDOEs_LV} and Eq.~\ref{eq:market-clear} may become intractable in practice. Possible practical resolution schemes include using convex approximations of AC power flow, where similar simplified resolution techniques exist \cite{Chau2018ACformulation}, or evaluating constraints with high probability by sampling.

\subsection{Economic and market design implications}

\subsubsection{Initial DOE allocation} the initial DOE allocation in SecuLEx is achieved by solving the lexicographic max-min problem in Eq.~\ref{eq:CalculateDOEs_LV}. It implements an egalitarian rule, where the sizes of the smallest envelopes are maximized first. Beyond the technical challenges discussed previously, the uniqueness of the solution is not guaranteed, and the fairness requirement may lead to inefficient operation. Typically, envelopes may be small due to blocking consumers, or may be unrelated to the effective energy needs. Alternative allocation rules exist to deal with such problems, including allocation for maximum total network utilization, priority allocation favoring certain customer classes (e.g., critical loads), or market mechanisms that maximize social welfare through an auction for initial DOEs. Similarly, we maximize the size of envelopes, measured in kilowatts. Alternative objectives may account for relative sizes depending on the nominal consumption of nodes or use different metrics to quantify the sizes of DOEs.

\subsubsection{Market choices} SecuLEx implements a continuous market. First, by definition, continuous markets clear orders instantaneously and continuously when they arrive. It allows for fast reallocation depending on evolving needs and is thus the best choice when operating close to real time under high uncertainty. It nevertheless requires intensive customer involvement, e.g., compared to single-auction. In addition, fast clearing may be difficult in practice due to computational delays. Second, continuous markets often implement a pay-as-bid settlement. This choice originates as the process does not provide an aggregated market price. It may, in practice, lead to allocation inefficiency and gaming opportunities compared to uniform or Vickrey pricing \cite{Fabra2006}.

\subsubsection{Market efficiency and liquidity} in principle, markets allow optimal allocation, but inefficiencies are possible. Typically, too few participants may be involved, resulting in limited liquidity and price volatility. Alternatively, participants with high flexibility could exercise market power through strategic bidding or withholding limits. Moreover, this work assumes zero transaction costs, while in real markets significant overheads (e.g., fees, communication, and administrative costs) exist. These can erode the benefits of small-scale trades and further reduce liquidity.

\subsubsection{Regulation and surplus treatment} implementing SecuLEx in practice requires, on the one hand, guaranteeing independence of the market operation and, on the other hand, guaranteeing non-discriminatory network access \cite{EU2019ElectricityDirective}. Concerning market independence, it requires introducing a market operator distinct from the DSO, together with clear regulation. 
This aspect is largely neglected in this work. 
In addition, we discussed that market inefficiencies may lead to discriminatory network access, which shall be controlled by an independent observer in practice. 
Finally, in its current form, SecuLEx leads to a market surplus. Regulation should ensure that it is used to reduce tariffs or reinvest in the grid and avoid providing incentives to underinvest in infrastructure.

%% file: Files/Conclusion.tex
\section{Conclusion} 
\label{conclusion}
This paper presented Secure Limit Exchange (SecuLEx), a new market-based paradigm to allocate power injection and withdrawal limits, called dynamic operating envelopes (DOEs), that guarantee network security during time periods.
Under this paradigm, distribution system operators (DSOs) assign initial DOEs to customers that can be exchanged afterward through a market, while preserving security guarantees.
We formulated the DOE allocation and market-clearing problems as tractable optimization problems under DC power flow and radial network assumptions, and proved that security can then be efficiently verified.
In our illustrative low-voltage (LV) network example, DOE trading showed potential for reducing renewable curtailment and improving network utilization and social welfare relative to alternative management schemes.
Overall, SecuLEx demonstrates that envelope trading can extract more value from existing infrastructure and provide incentives for flexibility without requiring central control. Scalability to larger systems, extending the framework beyond the simplified DC power flow model to a more realistic AC model that can accurately manage voltage levels and reactive power, alternative allocation and pricing rules, and the regulatory treatment of market social welfare, remains important directions for future work.

%% file: Files/ProofSafeness_DC.tex
\subsection{Security function proof}
\label{app:proof}
This appendix provides proof that, under radial network and DC assumptions, branch currents are monotone non-decreasing and node voltages are constant as a function of customer consumption.
As a result, if the system is secure at the two extreme points of the DOE matrix, it is also secure for every power respecting the limits.

\paragraph{Network description} as described in Section~\ref{ss:network}, we model a distribution network as a graph $\mathfrak{G} = (\mathcal{N}, \mathcal{E})$ with nodes $\mathcal{N}$, edges $\mathcal{E}$, and customers $\mathcal{C} \subset \mathcal{N}$.
The radial topology ensures the existence of a unique root node $r \in \mathcal{N}$ representing the substation or transformer connecting the feeder to the upstream grid. We subsequently assume that edges are oriented from root to leaves.
For each node $n$, we define the set of downstream nodes $\mathcal{D}_n$ as the set of nodes $m$ for which the unique path from $m$ to $r$ passes by $n$.

Each edge $e \in \mathcal{E}$ has an associated resistance $R_{e}$ and reactance $X_{e}$. We denote by $I_{e}$ the line current, by $P_{e}$ the active power flow, and by $Q_{e}$ the reactive power flow.
To each node $n \in \mathcal{N}$ corresponds an active power consumption $P_n$, and we denote by $V_n$ the voltage magnitude.

Customer power limits are represented by the matrix $\mathbf{L} = [\underline{P}, \overline{P}] \in \mathbb{R}^{|\mathcal{C}| \times 2}$, where $[\underline{P}_c, \overline{P}_c]$ are the bounds for each customer $c \in \mathcal{C}$. 
The function $VerifyLimits(\mathfrak{G}, \mathbf{L})$ checks that the network is secure for all $P \in \mathbf{L}$. 

\paragraph{Assumptions} we make the following assumptions on the network topology and physics linking electrical values.
\begin{enumerate}[label=\textbf{A\arabic*}, leftmargin=1.5cm]
    \item \textit{Radial topology.} The network graph is a tree. \label{ass:radial}
    \item \textit{DC load flow.} Balanced steady-state, with unity and flat voltages ($\forall n \in \mathcal{N} : V_n = V_r = 1$p.u.), lossless lines ($\forall e \in \mathcal{E} : R_{e} = 0$), and no reactive power ($\forall e \in \mathcal{E} : Q_{e} = 0$). \label{ass:dc}
\end{enumerate}

\begin{theorem}[Monotonicity of current]
\label{thm:dc_monotone}
Under Assumptions \ref{ass:radial}–\ref{ass:dc}, the branch current $I_{e}$ is a monotone non-decreasing function of consumption $P_c$:
\[
    \frac{\partial I_{e}}{\partial P_c} \ge 0 \; , 
    \quad \forall e\in \mathcal{E}, \forall c \in \mathcal{C} \; .
\]

\end{theorem}

\begin{proof}
Under assumptions \ref{ass:radial} and \ref{ass:dc}, the active power flowing in each branch $e \in \mathcal{E}$ respects the balance:
\begin{align}
\label{eqp:active_power_balance}
    P_{e}  &= \sum_{k \in \mathcal{D}_e} P_k \; .
\end{align}
Furthermore, in each branch $e \in \mathcal{E}$:
\begin{align}
\label{eqp:ohm_law}
    I_{e} = \frac{1}{V_r} P_{e} \; .
\end{align}
Combining Eq. \ref{eqp:active_power_balance} and  Eq. \ref{eqp:ohm_law}, we obtain the relationship between the branch current and power consumption:
\begin{align}
    I_{e}  &= \frac{1}{V_r} \left( \sum_{k \in \mathcal{D}_e} P_k \right ) \; .
\end{align}
Let us differentiate the current with respect to any downstream node $n \in \mathcal{D}_e$:
\begin{align}
    \frac{\partial I_{e}}{\partial P_n}
    = \frac{1}{V_r} \frac{\partial }{\partial P_n} \left( \sum_{k \in \mathcal{D}_e} P_k \right )
    = \frac{1}{V_r} \geq 0 \; .
\end{align}
Let us differentiate the current with respect to any other node $n \notin \mathcal{D}_e$ that is not the root $r$:
\begin{align}
    \frac{\partial I_{e}}{\partial P_n} = 0 \; .
\end{align}
The branch currents are therefore monotone non-decreasing with respect to any consumption.

\end{proof}



\begin{theorem}[Boundary check constraints using monotonicity] \label{thm:boundary_check} Under Assumptions \ref{ass:radial}–\ref{ass:dc}, for any $\mathbf{L} = [\underline{P}, \overline{P}] \in \mathbb{R}^{|\mathcal{C}| \times 2}$, if both current and voltage security constraints are respected for $\underline{P}$ and $\overline{P}$, then the security constraints are also respected for each $P \in \mathbf{L}$.
\end{theorem}

\begin{proof}
Under assumptions \ref{ass:radial} and \ref{ass:dc}, the voltage at each node is constant. Therefore, if the voltage security constraints are respected for one power profile, it is respected for each power profile.

Under assumptions \ref{ass:radial} and \ref{ass:dc}, Theorem~\ref{thm:dc_monotone} holds, and each branch current $I_{ij}$ is a monotone non-decreasing function of all customer consumptions. This implies that the maximum current and minimum current (i.e., maximum negated current) as a function of $P$ over the domain $P \in \mathbf{L}$ is achieved either in $\underline{P}$ or in $\overline{P}$. Therefore, if the current security constraints are respected for $\underline{P}$ and $\overline{P}$, then they are also respected $\forall P \in \mathbf{L}$.

\end{proof}

According to Theorem \ref{thm:boundary_check}, the $VerifyLimits$ function can be implemented as a function verifying that security is ensured at the two DOE boundary profiles $\underline{P}$ and $\overline{P}$. 
This makes the verification computationally efficient while guaranteeing security under the DC approximation.